\documentclass[11pt]{amsart}
\usepackage{amssymb,mathtools,bm,graphicx}
\usepackage{booktabs}
\usepackage{enumitem}
\usepackage{pgfplots}
\usepackage{pgfplotstable}
\pgfplotsset{compat=1.18}
\pgfplotsset{
  every axis/.append style={
    legend columns=-1,
    legend cell align={center},
    legend style={at={(0.5,1.16)}, anchor=south, font=\scriptsize, draw=none, fill=none, /tikz/every even column/.append style={column sep=0.45em}},
    label style={font=\small},
    tick label style={font=\scriptsize}
  }
}
\usepgfplotslibrary{groupplots}
\renewcommand{\AA}{\bm A}
\newcommand{\BB}{\bm B}
\newcommand{\GG}{\bm G}
\newcommand{\HH}{\bm H}
\newcommand{\Id}{\bm I}
\newcommand{\QQ}{\mathcal Q}
\newcommand{\RR}{\mathcal R}
\newcommand{\EE}{\mathcal E}
\newcommand{\calD}{\mathcal D}
\newcommand{\calK}{\mathcal K}
\newcommand{\calO}{\mathcal O}
\newcommand{\uu}{\bm u}
\newcommand{\Wi}{\mathrm{Wi}}
\newcommand{\dd}{\,\mathrm d}
\newcommand{\dt}{\Delta t}
\newcommand{\SPD}{\mathbb S_{++}}
\newcommand{\tr}{\operatorname{tr}}

\newcommand{\Log}{\operatorname{Log}}
\newcommand{\Exp}{\operatorname{Exp}}
\newcommand{\divv}{\nabla\!\cdot}
\newcommand{\grad}{\nabla}
\newcommand{\norm}[1]{\left\|#1\right\|}

\makeatletter
\newcounter{manualresult}[section]
\renewcommand{\themanualresult}{\thesection.\arabic{manualresult}}
\newcommand{\manualresulthead}[2]{%
  \par\medskip\stepcounter{manualresult}%
  \protected@edef\@currentlabel{\themanualresult}%
  \noindent{\normalfont\scshape #1~\themanualresult.} #2\ }
\newenvironment{theorem}{\manualresulthead{Theorem}{\itshape}}{\par\medskip}
\newenvironment{lemma}{\manualresulthead{Lemma}{\itshape}}{\par\medskip}
\newenvironment{proposition}{\manualresulthead{Proposition}{\itshape}}{\par\medskip}
\newenvironment{definition}{\manualresulthead{Definition}{\normalfont}}{\par\medskip}
\renewenvironment{proof}{\par\medskip\noindent{\itshape Proof.}\ }{\hfill$\square$\par\medskip}
\makeatother

\title[Entropy-Compatible Reconstruction]{Entropy-Compatible Reconstruction for High-Weissenberg Viscoelastic Flow}
\author{Sai Peng}
\address{School of Mathematics and Computational Science, Xiangtan University, Xiangtan, Hunan, China}
\email{pscfd@xtu.edu.cn}

\begin{document}

\begin{abstract}
Log-conformation and square-root reconstructions preserve positive definiteness in high-Weissenberg viscoelastic simulations, but positivity alone does not guarantee compatibility with the discrete free-energy balance.  We identify three reconstruction-level mechanisms by which strictly positive tensors can still generate nonphysical behavior: Jensen-type entropy bias, exponential amplification of logarithmic perturbations in highly stretched states, and sign-indefinite polymeric-work defects caused by using incompatible tensors in stress work and entropy variables.  We formulate an entropy-compatible reconstruction principle and a corrected logarithmic reconstruction selected by a least-damping entropy constraint.  The correction is local, positive, computable by bisection, spectrally controlled, and compatible with coupled velocity--pressure--conformation time stepping.  We prove existence of the maximal admissible parameter, convexity of the entropy profile along the logarithmic path, a compatible free-energy estimate, a defect-budget estimate for noncompatible reconstructions, asymptotic inactivity on high-order admissible defects, and a conditional high-stretch resolution advantage in log-relative and entropy metrics.  Reproducible diagnostics compare logarithmic, square-root, and linear reconstructions and verify the predicted entropy defects, work defects, stress-force errors, and high-Weissenberg accumulation.
\end{abstract}
\maketitle

\par\medskip\noindent\textbf{Keywords.} 
Oldroyd--B, log-conformation, square-root conformation, high Weissenberg number, entropy stability, positive definite reconstruction
\par\medskip

\par\medskip\noindent\textbf{MSC 2020.} 
65M08, 65M12, 76A10, 76M20
\par\medskip

\section{Introduction}

The high-Weissenberg number problem remains a central obstacle in viscoelastic flow simulation \cite{BirdArmstrongHassager1987,Keunings1986}.  In conformation-tensor models the polymeric state is a symmetric positive definite tensor $\AA$.  Positive definiteness is indispensable: without it the elastic free energy, $\log\det\AA$, and the matrix logarithm are not meaningful.  Log-conformation methods enforce this constraint by writing
\[
  \AA=\Exp\Psi,\qquad \Psi=\Log\AA,
\]
and square-root methods enforce it by writing $\AA=\BB\BB^T$.  These are powerful and widely used ideas \cite{FattalKupferman2004,FattalKupferman2005,HulsenFattalKupferman2005,BalciThomasesRenardy2011}.  However, numerical experience shows that high-Weissenberg computations may still develop nonphysical stress overshoots or elastic-energy growth even when $\AA$ remains positive definite.

The reason is that positivity is a cone constraint, while free-energy compatibility is a coupled balance constraint.  The elastic entropy
\[
  \Phi(\AA)=\tr\AA-\log\det\AA-d
\]
is nonlinear in both $\AA$ and $\Psi$.  In addition, the continuous free-energy law relies on exact cancellation between the polymeric work in the momentum equation and the stretching work in the conformation equation.  If these terms use different reconstructed tensors, a sign-indefinite work defect remains.

This paper studies that reconstruction-level issue.  We show that a positive high-order reconstruction can be nonphysical in three ways: it can introduce nonlinear entropy bias, amplify logarithmic perturbations through the matrix exponential, or leave an uncancelled coupling work defect.  We then define an entropy-compatible corrected logarithmic reconstruction.  The correction accepts the raw high-order logarithmic reconstruction if it satisfies a local entropy budget; otherwise it moves back along the logarithmic segment from a physical predictor toward the raw reconstruction and chooses the largest admissible parameter.

The contribution is not another positivity parametrization.  Rather, it is an admissibility principle for positive reconstructions before they enter the coupled stress work.  The corrected reconstruction is least damping, local or cellwise, positive by construction, bisection computable, and asymptotically inactive on sufficiently resolved high-order defects.

\subsection{Related work and positioning}

The high-Weissenberg number problem has been recognized for decades as a central obstruction in viscoelastic computation \cite{Keunings1986}.  Log-conformation methods \cite{FattalKupferman2004,FattalKupferman2005,HulsenFattalKupferman2005} and square-root or factorized methods \cite{BalciThomasesRenardy2011} address the most visible part of the problem: loss of positive definiteness.  Free-energy-dissipative methods address a complementary part: compatibility with the thermodynamic balance \cite{BoyavalLelievreMangoubi2009}.

The point of this paper is the gap between these two ideas.  A reconstruction can be strictly positive and still alter the entropy or the stress-work cancellation before the time integrator sees the state.  This is particularly damaging at high Weissenberg number, where weak relaxation allows small reconstruction defects to persist and accumulate.  The corrected log reconstruction proposed here is therefore not a replacement for log-conformation methods; it is an entropy-compatibility layer placed on top of a positive reconstruction.

The main results are used in the following order.  First we identify reconstruction-level defects: Jensen entropy bias, exponential amplification, and sign-indefinite work defects.  Then we define a logarithmic path from a physical predictor to a raw high-order reconstruction and choose the largest point satisfying an entropy budget.  Finally we prove that this correction is positive, spectrally controlled, least damping on the path, compatible with a coupled energy estimate, and asymptotically inactive when the raw high-order defect is already entropy-admissible.

\begin{table}[htbp]
\centering
\caption{Main claims, analytical mechanisms, and numerical evidence.}
\label{tab:claim-evidence-log}
\small
\begin{tabular}{p{0.25\linewidth}p{0.36\linewidth}p{0.29\linewidth}}
\toprule
Claim & Analytical mechanism & Main evidence \\
\midrule
Positivity is not enough & nonlinear maps and mismatched work tensors can add entropy or leave sign-indefinite defects & scalar bias, square-root comparison, work-defect table \\
Log defects amplify at high stretch & the Frechet derivative of $\Exp$ scales with $\lambda_{\max}(\AA)$ & matrix amplification and coupled stress-force diagnostics \\
The correction is least damping & choose the largest entropy-admissible point on the logarithmic segment & bisection construction and entropy-correction diagnostic \\
Coupled energy is controlled & use the accepted tensor consistently in stress work, stretching, entropy variables, and quadrature & discrete energy estimate and dynamic high-$\Wi$ loop \\
High-order behavior is retained & mesh-scaled budgets make compatible high-order defects inactive or asymptotically small & resolution and stretch sweeps \\
\bottomrule
\end{tabular}
\end{table}

The paper isolates the reconstruction layer and its interaction with a coupled velocity--pressure--conformation step.  Optimal nonlinear preconditioners, adaptive stress-layer resolution, and full three-dimensional benchmark suites are important solver questions, but the mathematical issue addressed here is sharper: which positive reconstructed tensor may enter the coupled entropy and stress-work balance.

For this reason, the paper focuses on the entropy budget, the coupled work cancellation, the high-order inactivity estimate, and the resolution mechanism.

\section{Model entropy and positive reconstructions}

We use the conformation form of an incompressible viscoelastic model,
\begin{align}
  \partial_t\uu+\uu\cdot\grad\uu-\beta\Delta\uu+\grad p
  &=\frac{1-\beta}{\Wi}\divv(\AA-\Id),                         \label{eq:mom}\\
  \divv\uu&=0,                                                   \label{eq:div}\\
  \partial_t\AA+\uu\cdot\grad\AA
  -(\grad\uu)\AA-\AA(\grad\uu)^T
  &=-\frac1{\Wi}(\AA-\Id)+\varepsilon\Delta\AA .                \label{eq:conf}
\end{align}
The entropy variable is
\[
  \HH(\AA)=D\Phi(\AA)=\Id-\AA^{-1}.
\]
Testing the momentum equation by $\uu$ and the conformation equation by $\HH(\AA)$ gives the coupling cancellation
\[
  \frac{1-\beta}{\Wi}\int_\Omega(\AA-\Id):\grad\uu\,\dd x
  -
  \frac{1-\beta}{2\Wi}\int_\Omega
  \big((\grad\uu)\AA+\AA(\grad\uu)^T\big):(\Id-\AA^{-1})\,\dd x
  =0 .
\]
Thus a reconstruction must not only produce a positive tensor; the same accepted tensor must be used in stress work, stretching work, entropy variables, and entropy quadrature.

Let $\RR_h$ be a high-order reconstruction to quadrature points.  Three representative choices are
\[
  \AA_q=\RR_h\AA,\qquad
  \AA_q=\Exp(\RR_h\Log\AA),\qquad
  \AA_q=(\RR_h\BB)(\RR_h\BB)^T .
\]
The last two enforce positivity.  The next section explains why this is not enough.

\subsection{Relative entropy on spectral sets}

For $\AA,\BB\in\SPD$ define the relative matrix entropy
\[
  \Phi(\AA\mid\BB)
  =
  \Phi(\AA)-\Phi(\BB)
  -(\Id-\BB^{-1}):(\AA-\BB).
\]
It is nonnegative by convexity and vanishes only when $\AA=\BB$.  On a compact spectral set
\[
  \calK_{m,M}=\{\AA\in\SPD:\;m\Id\le\AA\le M\Id\},
\]
relative entropy is equivalent to the squared logarithmic distance:
\[
  c_{m,M}\norm{\Log\AA-\Log\BB}_F^2
  \le
  \Phi(\AA\mid\BB)
  \le
  C_{m,M}\norm{\Log\AA-\Log\BB}_F^2 .
\]
Indeed, the logarithm is a smooth diffeomorphism on $\calK_{m,M}$, and the Hessian of $\Phi$ is positive definite there:
\[
  D^2\Phi(\AA)[E,E]=\tr(\AA^{-1}E\AA^{-1}E).
\]
This observation is the metric bridge used later: the correction is defined by entropy, but its accuracy can be measured in logarithmic variables on resolved spectral sets.

\begin{proposition}
At a quadrature point let $\GG$ be a trace-free discrete velocity gradient.  If the momentum equation uses $\AA_m$ in the polymeric stress work while the conformation entropy equation uses $\AA_e$ in the stretching term, then the uncancelled coupling contribution is
\[
  \frac{1-\beta}{\Wi}(\AA_m-\AA_e):\GG .
\]
Exact cancellation for all trace-free $\GG$ holds when the deviatoric parts agree, in particular when $\AA_m=\AA_e$.
\end{proposition}

\begin{proof}
The stretching identity gives
\[
  \frac12(\GG\AA_e+\AA_e\GG^T):(\Id-\AA_e^{-1})
  =
  (\AA_e-\Id):\GG .
\]
The momentum stress work contains $(\AA_m-\Id):\GG$.  Subtracting the two expressions leaves the displayed defect.
\end{proof}

\section{Why positive reconstruction can be nonphysical}

\begin{proposition}
Let $a=\exp\psi$ and $\phi(a)=a-\log a-1$.  If $\psi=\psi_0+\eta$ with zero mean perturbation $\eta$, then the reconstructed physical mean and entropy satisfy, for nontrivial $\eta$,
\[
  \langle e^\psi\rangle>e^{\psi_0},\qquad
  \langle\phi(e^\psi)\rangle>\phi(e^{\psi_0})
\]
whenever the state is away from the entropy minimizer in the corresponding direction.
\end{proposition}

\begin{proof}
This is Jensen's inequality applied to the convex exponential map and to the convex entropy profile $\phi(e^\psi)$ on the relevant interval.  A Taylor expansion gives the leading positive bias proportional to $\langle\eta^2\rangle$.
\end{proof}

\begin{lemma}
Let $\AA=\Exp\Psi$ and perturb $\Psi$ by a small symmetric matrix $E$.  Then
\[
  \Exp(\Psi+E)-\Exp\Psi
  =
  \int_0^1 \Exp((1-s)\Psi)E\Exp(s\Psi)\,\dd s+\calO(\norm{E}^2).
\]
Consequently, the induced perturbation in $\AA$ is amplified by the local stretch scale, and stress-force defects grow with $\lambda_{\max}(\AA)$.
\end{lemma}

\begin{proof}
The formula is the Frechet derivative of the matrix exponential.  The norm bound follows by estimating the exponential factors with the largest eigenvalue of $\AA$.
\end{proof}

\begin{proposition}
Suppose the tensor used in the momentum stress work is $\AA_m$ while the tensor used in the conformation entropy/stretching term is $\AA_e$.  Both may be positive definite.  The uncancelled polymeric-work defect contains
\[
  \frac{1-\beta}{\Wi}
  \int_\Omega(\AA_m-\AA_e):\grad\uu\,\dd x ,
\]
which has no fixed sign.
\end{proposition}

\begin{proof}
Repeat the energy calculation with $\AA_m$ in the momentum stress and $\AA_e$ in the stretching/entropy term.  The cancellation is exact only when the two tensors agree.  The displayed residual remains and can inject or remove energy depending on its sign.
\end{proof}

These results show that log or square-root positivity is not a sufficient criterion for high-Weissenberg robustness.

\section{Entropy-compatible corrected log reconstruction}

Let $\QQ_h=\{(x_q,w_q)\}$ be the quadrature rule.  A physical predictor and a raw high-order logarithmic reconstruction are denoted by
\[
  \widehat\AA_q=\Exp\widehat\Psi_q,\qquad
  \widetilde\AA_q=\Exp\widetilde\Psi_q .
\]
For $0\le\theta\le1$ define the logarithmic path
\[
  \AA_q(\theta)
  =
  \Exp\!\big(\widehat\Psi_q+\theta(\widetilde\Psi_q-\widehat\Psi_q)\big).
\]
The corrected reconstruction chooses the largest $\theta$ such that
\[
  \sum_qw_q\Phi(\AA_q(\theta))
  \le
  \sum_qw_q\Phi(\widehat\AA_q)+\tau_h ,
  \label{eq:budget}
\]
where $\tau_h\ge0$ is the entropy budget.  If the raw reconstruction satisfies the budget, then $\theta=1$ and nothing is changed.

The implementation is a scalar bisection problem:
\begin{enumerate}[label=\textup{(\arabic*)},leftmargin=*]
\item Compute the physical entropy $J(0)$ and raw entropy $J(1)$.
\item If $J(1)\le J(0)+\tau_h$, accept the raw reconstruction.
\item Otherwise bracket $\theta_\star$ in $[0,1]$ and bisect until the remaining entropy gap or bracket length is below the requested tolerance.
\item Use the accepted tensor in every stress, stretching, entropy, and quadrature term.
\end{enumerate}
The same rule can be applied cellwise with local budgets $\tau_K$.  Cellwise selection avoids damping an entire mesh because of a localized stress layer, while the global entropy estimate follows after summing the cell budgets.

\subsection{Stopping criteria and budgets}

The bisection loop can be stopped by a bracket criterion or by an entropy-gap criterion.  If $J$ is Lipschitz on the interval with constant $L_J$, then a bracket of length $\delta_\theta$ gives an entropy uncertainty at most $L_J\delta_\theta$.  On compact spectral sets, $L_J$ is bounded by the endpoint logarithmic defect and the maximum conformation stretch.  Hence a fixed bisection depth can be chosen so that the search error is below the discretization budget.

For mesh-dependent calculations we use budgets of the form
\[
  \tau_h=C_\tau h^{2k+2}
  \quad\hbox{or}\quad
  \tau_K=C_\tau h_K^{2k+2}|K|.
\]
This scaling matches the squared size of a $(k+1)$st-order reconstruction defect in an entropy metric.  If the raw reconstruction is already compatible to that order, the correction is inactive.  If it is not, the correction removes only the entropy-incompatible component.  This is the sense in which the method is a high-order admissibility filter rather than a low-order limiter.

In floating point arithmetic the acceptance test should include a guard:
\[
  J(\theta)\le J(0)+\tau_h-\epsilon_{\rm eval}.
\]
Here $\epsilon_{\rm eval}$ bounds the accumulated error in evaluating exponentials, determinants, and quadrature sums.  This prevents a roundoff-level violation of the budget from being interpreted as a physical entropy increase.

\subsection{Structural invariances}

The correction preserves the tensorial information that is usually needed by a high-order implementation.  First, it is frame indifferent.  If the endpoint logarithms are transformed by an orthogonal matrix $Q$, then
\[
  \Exp\!\big(Q^T\Psi Q\big)=Q^T\Exp(\Psi)Q,
  \qquad
  \Phi(Q^TAQ)=\Phi(A).
\]
Therefore the accepted parameter $\theta_\star$ is unchanged by a rigid change of tensor basis, and the accepted conformation tensor transforms covariantly.

Second, if the raw logarithmic reconstruction defect has zero cell average,
\[
  \sum_{q\in K}w_q(\widetilde\Psi_q-\widehat\Psi_q)=0,
\]
then a cellwise correction using a single $\theta_K$ preserves that zero log-moment:
\[
  \sum_{q\in K}w_q
  \big(\widehat\Psi_q+\theta_K(\widetilde\Psi_q-\widehat\Psi_q)
  -\widehat\Psi_q\big)=0.
\]
Thus the correction damps high-frequency incompatible modes without shifting the cell average in logarithmic variables.  This property is useful when the underlying reconstruction is conservative or moment-preserving in log space.

Third, the correction is monotone with respect to the budget.  If $\tau_h^{(1)}\le\tau_h^{(2)}$, then the corresponding admissible sets satisfy
\[
  \{\theta:J(\theta)\le J(0)+\tau_h^{(1)}\}
  \subset
  \{\theta:J(\theta)\le J(0)+\tau_h^{(2)}\},
\]
so the largest admissible parameter cannot decrease when the budget is relaxed.  This gives a transparent way to tune the method from strict no-extra-entropy reconstruction to high-order permissive reconstruction.

\subsection{Relation to square-root reconstruction}

The same admissibility principle can be applied to square-root variables.  If $\AA=\BB\BB^T$ and a physical predictor $\widehat\BB$ and raw reconstruction $\widetilde\BB$ are available, one may consider
\[
  \BB(\theta)=\widehat\BB+\theta(\widetilde\BB-\widehat\BB),
  \qquad
  \AA(\theta)=\BB(\theta)\BB(\theta)^T ,
\]
and select the largest $\theta$ satisfying the same entropy budget.  This gives positivity whenever $\BB(\theta)$ remains nonsingular, and it provides the same reconstruction-level entropy control.

There are two reasons the logarithmic path is emphasized in this paper.  First, $\Exp$ maps every symmetric logarithm to $\SPD$, so positivity does not require a separate nonsingularity check along the segment.  Second, in high-stretch regimes the smooth variable is often closer to $\Psi=\Log\AA$ than to either $\AA$ or a chosen square root.  The log path therefore gives a natural way to preserve spectral envelopes and measure relative error.  The numerical comparisons with square-root reconstruction are not meant to dismiss square-root methods; they show that positivity through any nonlinear map can introduce entropy bias unless an entropy compatibility check is added.

The square-root variant is still useful in implementations where factor variables are already evolved.  In that case the corrected reconstruction should use the same principle: define a physical predictor, define the raw high-order defect, choose the largest admissible parameter with respect to $\Phi(\AA)$, and use the accepted tensor consistently in stress work and entropy variables.  The analysis of the work defect is unchanged because it depends only on the final tensors entering the coupled balance.

\begin{definition}
A reconstruction is entropy-compatible if the accepted tensor satisfies \eqref{eq:budget} and the same accepted tensor is used in the stress work, stretching term, entropy variable, and entropy quadrature.
\end{definition}

\begin{proposition}
The function
\[
  J(\theta)=\sum_qw_q\Phi(\AA_q(\theta))
\]
is convex on $[0,1]$.  Hence the admissible set $\{\theta:J(\theta)\le J(0)+\tau_h\}$ is an interval containing zero, and the maximal admissible parameter $\theta_\star$ exists.  Bisection returns a positive tensor satisfying the entropy budget.
\end{proposition}

\begin{proof}
For each quadrature point, $\theta\mapsto\Phi(\Exp(\widehat\Psi_q+\theta\delta\Psi_q))$ is convex because the trace-exponential part is convex and the log-determinant contribution is affine in the logarithmic variable.  Positive quadrature preserves convexity.  The interval property and bisection follow immediately.
\end{proof}

\begin{proposition}
Among all admissible states on the logarithmic segment, $\theta_\star$ minimizes the logarithmic distance to the raw reconstruction.  Thus the correction is least damping on that path.
\end{proposition}

\begin{proof}
The logarithmic distance from $\AA(\theta)$ to the raw endpoint is proportional to $(1-\theta)\norm{\widetilde\Psi-\widehat\Psi}$.  Maximizing admissible $\theta$ is therefore equivalent to minimizing this distance.
\end{proof}

\begin{proposition}
If $\alpha\Id\le\widehat\Psi_q,\widetilde\Psi_q\le\beta\Id$ for all $q$, then every corrected state satisfies
\[
  e^\alpha\Id\le\AA_q(\theta)\le e^\beta\Id .
  \]
The construction is also invariant under orthogonal changes of frame.
\end{proposition}

\begin{proof}
The logarithmic path is a convex combination of endpoint logarithms, so its spectrum remains in $[\alpha,\beta]$.  Exponentiation gives the conformation bounds.  Orthogonal equivariance follows from $\Exp(Q^TYQ)=Q^T\Exp(Y)Q$ and the invariance of $\Phi$.
\end{proof}

\begin{proposition}
Suppose the reconstruction defect satisfies
\[
  \sum_q w_q\,D\Phi(\widehat\AA_q):
  D\Exp_{\widehat\Psi_q}(\widetilde\Psi_q-\widehat\Psi_q)
  =\calO(h^{2k+2})
\]
and $\norm{\widetilde\Psi-\widehat\Psi}_{\QQ_h}=\calO(h^{k+1})$ on a compact spectral set.  If $\tau_h$ is chosen of order $h^{2k+2}$, then the correction is inactive or asymptotically small; more precisely, the accepted state preserves the same log-relative order as the raw reconstruction.
\end{proposition}

\begin{proof}
Taylor expand $J(\theta)$ at $\theta=0$.  The stated first-variation condition makes the linear entropy excess no larger than the budget scale, and the second derivative is uniformly bounded on the spectral set.  Hence $J(1)-J(0)=\calO(h^{2k+2})$.  If the budget dominates this constant, $\theta_\star=1$; otherwise the bisection point satisfies $1-\theta_\star=\calO(h^{k+1})$, which does not change the formal logarithmic order.
\end{proof}

\section{Coupled energy estimate and accuracy}

Consider a time step in which the accepted tensor $\AA_q^{n+1}$ from the previous section is used consistently in the momentum stress, stretching term, entropy variable, and quadrature evaluation.  Define
\[
  \EE_h^{n}
  =
  \frac12\norm{\uu_h^n}_{L^2}^2
  +\frac{1-\beta}{2\Wi}\sum_qw_q\Phi(\AA_q^n).
\]

\begin{theorem}
For an entropy-compatible reconstruction, any admissible backward-Euler coupled update satisfies
\[
  \EE_h^{n+1}-\EE_h^n+\dt\,\calD_h^{n+1}
  \le \dt\,(\bm f^{n+1},\uu_h^{n+1})
  +\frac{1-\beta}{2\Wi}\tau_h ,
\]
where $\calD_h^{n+1}$ contains solvent, relaxation, and diffusion dissipations.
\end{theorem}

\begin{proof}
The proof is the discrete free-energy calculation.  The only additional term relative to the exact entropy-compatible case is the allowed reconstruction budget $\tau_h$.  Because the same accepted tensor is used in all coupling terms, the polymeric work cancellation is algebraic.
\end{proof}

Thus $\tau_h$ has a precise interpretation: it is not artificial viscosity and not a hidden time-step restriction.  It is the maximum elastic entropy that the reconstruction stage is allowed to add beyond the physical predictor.  Choosing $\tau_h=0$ gives a strict no-extra-entropy reconstruction; choosing $\tau_h=\calO(h^{2k+2})$ allows high-order admissible defects while still keeping the reconstruction error below the target accuracy.

\begin{theorem}
If a noncompatible scheme uses different positive tensors $\AA_m$ and $\AA_e$ in stress work and entropy/stretching terms, then its energy balance contains a defect bounded by
\[
  C\,\norm{\AA_m-\AA_e}_{\QQ_h}\norm{\grad\uu_h}_{\QQ_h}.
\]
\end{theorem}

\begin{proof}
The uncancelled term is the quadrature version of the work defect identified above.  Cauchy's inequality gives the bound.
\end{proof}

\begin{theorem}
Assume a compact spectral set, bounded quadrature weights, and a high-order logarithmic reconstruction defect satisfying a first-variation consistency condition.  With a mesh-scaled budget $\tau_h=\calO(h^{2k+2})$, the corrected reconstruction is asymptotically inactive: $\theta_\star=1$ for sufficiently small $h$, or $1-\theta_\star=\calO(h^{k+1})$ in the marginal case.  The corrected state preserves the formal log-relative accuracy of the raw high-order reconstruction.
\end{theorem}

\begin{proof}
Expand $J(\theta)$ around $\theta=0$.  The first variation vanishes or is of higher order by the consistency condition; the second variation is bounded on the compact spectral set.  Hence the entropy excess of the raw high-order defect is of the same order as the budget, so the correction becomes inactive or asymptotically small.
\end{proof}

\begin{theorem}
Let $\AA=\Exp\Psi$ with smooth $\Psi$ on a cell and $\lambda_{\max}(\AA)\le\Lambda$.  A direct reconstruction of $\AA$ sees derivatives that scale like $\Lambda$ in high-stretch layers, while logarithmic reconstruction resolves the smooth variable $\Psi$.  The corrected log reconstruction retains the log-relative order and satisfies the entropy budget.  Thus, in log-relative and entropy-compatible metrics, the method has a conditional high-stretch resolution advantage.
\end{theorem}

\begin{proof}
Approximation in logarithmic variables gives $\norm{\RR_h\Psi-\Psi}=\calO(h^{k+1})$.  Direct approximation of $\AA=\Exp\Psi$ inherits derivatives of the exponential and therefore scales with the local stretch.  The correction stays on the same logarithmic segment and cannot enlarge the log defect beyond the raw endpoint.
\end{proof}

The resolution statement is formulated in the norm in which the logarithmic variable is the resolved smooth object.  In absolute stress-force norms, both physical-space and logarithmic reconstructions can still feel the factor $\lambda_{\max}(\AA)$ after mapping back to $\AA$.  The advantage proved here is more precise: if the smooth object is $\Psi=\Log\AA$, then the corrected log reconstruction keeps the approximation high order in log-relative and entropy-compatible metrics, and it prevents the entropy-incompatible part of the mapped stress force from entering the coupled work balance.

To see the distinction, write $\AA=\Exp\Psi$.  A direct reconstruction of $\AA$ must approximate derivatives of $\Exp\Psi$, and the chain rule introduces factors of the form
\[
  D\Exp_{\Psi}[\partial_i\Psi]
  =
  \int_0^1 e^{(1-s)\Psi}(\partial_i\Psi)e^{s\Psi}\,\dd s .
\]
This term is bounded by $\lambda_{\max}(\AA)\norm{\partial_i\Psi}$.  In contrast, a logarithmic reconstruction approximates $\Psi$ directly.  Mapping back to $\AA$ still amplifies absolute physical-space errors, but the entropy-compatible correction ensures that the amplified part cannot appear as an uncontrolled positive entropy increment or uncancelled work defect.  This is exactly the high-$\Wi$ regime in which positivity alone is too weak a diagnostic.

\subsection{Production diagnostics}

For actual high-Weissenberg calculations, the correction is monitored through a small set of scalar diagnostics.  The first is the minimum eigenvalue, which detects loss of admissibility.  The second is the reconstruction entropy excess
\[
  \Delta_\Phi
  =
  \sum_qw_q\Phi(\AA_q^{\rm accepted})
  -
  \sum_qw_q\Phi(\widehat\AA_q),
\]
which measures how much elastic entropy the reconstruction stage adds beyond the physical predictor.  The third is the coupling defect
\[
  \mathcal W_h
  =
  \frac{1-\beta}{\Wi}
  \sum_qw_q(\AA_{m,q}-\AA_{e,q}):\nabla\uu_q ,
\]
which should vanish when the same accepted tensor is used consistently.

These quantities separate three different failure modes.  If $\lambda_{\min}$ becomes nonpositive, the conformation model itself has left its admissible set.  If $\Delta_\Phi$ is positive and large while eigenvalues remain positive, the reconstruction is adding artificial elastic energy.  If $\mathcal W_h$ is nonzero, the implementation is using incompatible tensors in the stress and entropy terms.  This diagnostic separation is one practical advantage of the present formulation: it tells the user whether the problem is positivity, entropy bias, or coupling mismatch.

\section{Numerical diagnostics}

The numerical section is organized to isolate the mechanisms above.  It includes scalar diagnostics, the matrix amplification test, the sign-defect test, a coupled velocity--stress manufactured diagnostic, a dynamic high-Weissenberg loop, and the mesh-and-stretch resolution benchmark.  These are the tests most directly connected to the analysis.

\begin{table}[htbp]
\centering
\caption{Diagnostics supporting the reconstruction theory.}
\label{tab:diagnostics}
{\small
\setlength{\tabcolsep}{3pt}
\begin{tabular}{p{0.24\linewidth}p{0.37\linewidth}p{0.25\linewidth}}
\toprule
Diagnostic & Purpose & Main observation \\
\midrule
Scalar log bias & Jensen entropy bias & positive entropy excess \\
Log/sqrt/linear comparison & positivity versus entropy bias & positivity alone is insufficient \\
Matrix exponential amplification & stretch-dependent perturbation growth & force defect grows with stretch \\
Work-defect sign & mismatch of stress and entropy tensors & defect has no fixed sign \\
Coupled manufactured test & stress-force defect in momentum equation & correction removes incompatible defect \\
Dynamic high-$\Wi$ loop & accumulation of positive log defects & raw log accumulates; corrected log controls \\
Resolution sweep & mesh and stretch dependence & corrected log controls entropy-compatible defect \\
\bottomrule
\end{tabular}
}
\end{table}

\subsection{Scalar and pointwise mechanisms}

The scalar log diagnostic uses $\psi(x)=\psi_0+\alpha\sin(2\pi x)$.  The perturbation has zero mean in logarithmic variables, but exponentiation creates a positive mean bias in $\AA$ and an entropy defect.  This is the simplest example showing that positivity-preserving nonlinear maps can create artificial elastic energy.

\begin{table}[htbp]
\centering
\caption{Scalar diagnostic for log reconstruction.  A zero-mean logarithmic fluctuation $\psi_0+\alpha\sin(2\pi x)$ is exponentiated with $\psi_0=3$.  Positivity is automatic, but the mean conformation and entropy are biased upward.}
\label{tab:scalar-log-bias}
\begin{tabular}{ccc}
\toprule
$\alpha$ & relative mean bias & entropy defect \\
\midrule
0.050000 & 6.250977E-04 & 1.255542E-02 \\
0.100000 & 2.501563E-03 & 5.024523E-02 \\
0.200000 & 1.002503E-02 & 2.013581E-01 \\
0.400000 & 4.040178E-02 & 8.114915E-01 \\
0.600000 & 9.204536E-02 & 1.848781E+00 \\
\bottomrule
\end{tabular}
\end{table}

The next diagnostic compares logarithmic, square-root, and linear reconstructions in a scalar setting.  Linear reconstruction has the smallest entropy defect while it stays positive, but it has no cone protection.  Logarithmic and square-root reconstructions preserve positivity for all tested amplitudes, but both introduce nonlinear entropy bias.

\begin{table}[htbp]
\centering
\scriptsize
\caption{Comparison of scalar positivity reconstructions around $A_0=20$.  Log reconstruction uses $A=A_0\exp(\alpha q)$, square-root reconstruction uses $A=A_0(1+\alpha q)^2$, and linear reconstruction uses $A=A_0(1+\alpha q)$ with $q=\sin(2\pi x)$.}
\label{tab:scalar-reconstruction-comparison}
\begin{tabular}{ccccc}
\toprule
$\alpha$ & log $\Delta\Phi$ & sqrt $\Delta\Phi$ & linear $\Delta\Phi$ & linear min \\
\midrule
0.050000 & 1.250195E-02 & 2.625117E-02 & 6.255868E-04 & 1.900000E+01 \\
0.100000 & 5.003126E-02 & 1.050189E-01 & 2.509427E-03 & 1.800000E+01 \\
0.200000 & 2.005006E-01 & 4.203068E-01 & 1.015342E-02 & 1.600000E+01 \\
0.400000 & 8.080356E-01 & 1.685277E+00 & 4.263868E-02 & 1.200000E+01 \\
0.600000 & 1.840907E+00 & 3.810721E+00 & 1.053605E-01 & 8.000004E+00 \\
\bottomrule
\end{tabular}
\end{table}

The matrix amplification diagnostic then fixes a small logarithmic perturbation and increases the background stretch.  The measured amplification grows proportionally to the largest conformation eigenvalue, as predicted by the Frechet derivative of the exponential map.

\begin{table}[htbp]
\centering
\caption{Amplification of a fixed symmetric logarithmic perturbation under the matrix exponential.  The diagnostic reports $\|\exp(\Psi+\varepsilon E)-\exp(\Psi)\|_F/(\varepsilon\|E\|_F)$ with $\varepsilon=10^{-5}$ and $\Psi=\operatorname{diag}(\log\lambda_{\max},0)$.}
\label{tab:matrix-exp-amplification}
\begin{tabular}{cc}
\toprule
$\lambda_{\max}(A)$ & amplification factor \\
\midrule
1.000000E+00 & 1.000003E+00 \\
1.000000E+01 & 7.539628E+00 \\
1.000000E+02 & 7.214908E+01 \\
1.000000E+03 & 7.136519E+02 \\
\bottomrule
\end{tabular}
\end{table}

Finally, the work-defect table verifies that two positive reconstructed tensors can produce opposite signs in the energy balance when the stress tensor and entropy tensor are not the same.

\begin{table}[htbp]
\centering
\caption{Sign-indefinite polymeric-work defects caused by using mismatched conformation reconstructions in the stress work and entropy variables.  The diagnostic uses the extensional gradient $\operatorname{diag}(1,-1)$; positive and negative log biases generate opposite signs.}
\label{tab:work-defect}
\begin{tabular}{ccc}
\toprule
$\alpha$ & positive defect & negative defect \\
\midrule
0.050000 & 3.567259E-03 & -3.565347E-03 \\
0.100000 & 1.430304E-02 & -1.427237E-02 \\
0.200000 & 5.776083E-02 & -5.726549E-02 \\
0.400000 & 2.401299E-01 & -2.318970E-01 \\
\bottomrule
\end{tabular}
\end{table}

\subsection{Coupled and dynamic diagnostics}

The coupled manufactured diagnostic uses a divergence-free periodic velocity and a smooth positive conformation field with prescribed stretch $\Lambda$.  Log and square-root reconstructions are positive, but their stress-force defects grow with stretch.  The corrected log reconstruction nearly removes the incompatible force component because the accepted tensor is tied to the entropy budget.

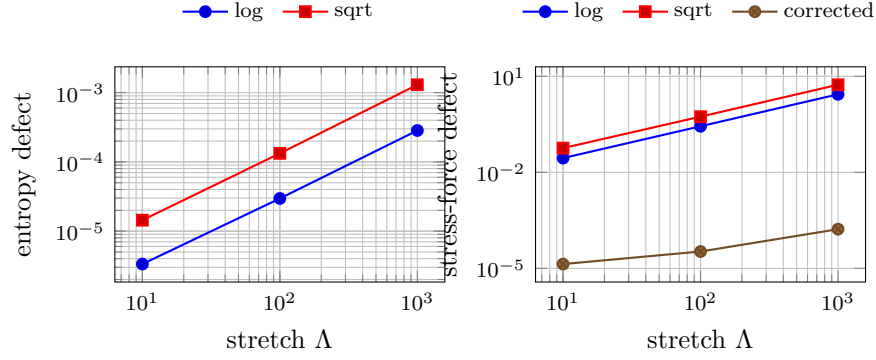
\begin{figure}[htbp]
\centering
\begin{tikzpicture}
\begin{groupplot}[
  group style={group size=2 by 1, horizontal sep=1.2cm},
  width=0.47\textwidth,
  height=0.35\textwidth,
  xmode=log,
  ymode=log,
  grid=both,
  xlabel={stretch $\Lambda$},
  legend style={at={(0.5,1.16)}, anchor=south, font=\scriptsize, draw=none, fill=none, /tikz/every even column/.append style={column sep=0.45em}},
  every axis plot/.append style={thick, mark=*}
]
\nextgroupplot[ylabel={entropy defect}]
\addplot table[x=stretch,y=log_entropy_defect,col sep=comma]{data/coupled_flow_diagnostic.csv};
\addlegendentry{log}
\addplot table[x=stretch,y=sqrt_entropy_defect,col sep=comma]{data/coupled_flow_diagnostic.csv};
\addlegendentry{sqrt}
\nextgroupplot[ylabel={stress-force defect}]
\addplot table[x=stretch,y=log_force_defect,col sep=comma]{data/coupled_flow_diagnostic.csv};
\addlegendentry{log}
\addplot table[x=stretch,y=sqrt_force_defect,col sep=comma]{data/coupled_flow_diagnostic.csv};
\addlegendentry{sqrt}
\addplot table[x=stretch,y=corrected_force_defect,col sep=comma]{data/coupled_flow_diagnostic.csv};
\addlegendentry{corrected}
\end{groupplot}
\end{tikzpicture}
\caption{Coupled manufactured diagnostic.  Positivity-preserving reconstructions generate entropy and stress-force defects that increase with stretch.  The corrected log reconstruction damps the incompatible defect before it enters the momentum equation.}
\label{fig:coupled-flow-diagnostic}
\end{figure}

\begin{table}[htbp]
\centering
\tiny
\caption{Coupled periodic velocity--stress diagnostic.  The exact velocity is divergence free and the exact conformation field is positive.  The table reports entropy defects, stress-force defects $\|\nabla\cdot(A_{\rm rec}-A_{\rm phys})\|_2$, and the damping factor selected by the entropy-compatible correction.}
\label{tab:coupled-flow-diagnostic}
\begin{tabular}{ccccccc}
\toprule
$\Lambda$ & $\theta$ & $\Delta\Phi_{\log}$ & $\Delta\Phi_{\sqrt{\ }}$ & $F_{\log}$ & $F_{\sqrt{\ }}$ & $F_{\rm corr}$ \\
\midrule
1.000000E+01 & 0.000488 & 3.341465E-06 & 1.439163E-05 & 2.782671E-02 & 5.709695E-02 & 1.358726E-05 \\
1.000000E+02 & 0.000122 & 2.963373E-05 & 1.328310E-04 & 2.744377E-01 & 5.482625E-01 & 3.350068E-05 \\
1.000000E+03 & 0.000061 & 2.845167E-04 & 1.303056E-03 & 2.736774E+00 & 5.446792E+00 & 1.670394E-04 \\
\bottomrule
\end{tabular}
\end{table}

The entropy-correction diagnostic follows a smooth physical predictor plus a high-frequency reconstruction defect.  Standard log reconstruction remains positive but accumulates entropy from the defect, whereas the corrected update returns to the predictor entropy within the prescribed budget.

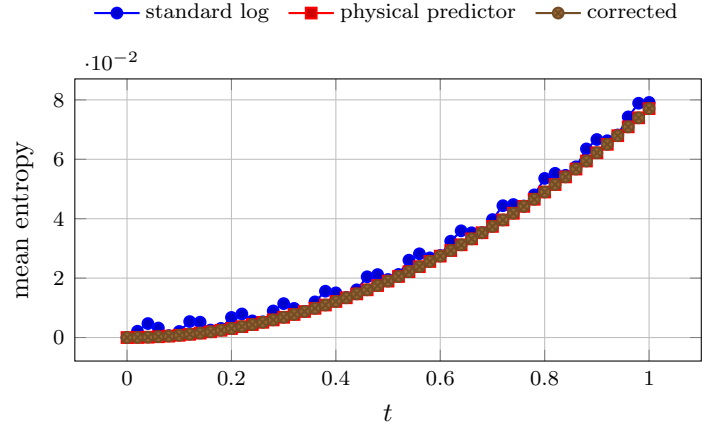
\begin{figure}[htbp]
\centering
\begin{tikzpicture}
\begin{axis}[
  width=0.78\textwidth,
  height=0.42\textwidth,
  grid=both,
  xlabel={$t$},
  ylabel={mean entropy},
  legend columns=-1,
  legend style={at={(0.5,1.16)}, anchor=south, font=\scriptsize, draw=none, fill=none, /tikz/every even column/.append style={column sep=0.45em}},
  every axis plot/.append style={thick}
]
\addplot table[x=t,y=standard_entropy,col sep=comma]{data/entropy_correction_history.csv};
\addlegendentry{standard log}
\addplot table[x=t,y=base_entropy,col sep=comma]{data/entropy_correction_history.csv};
\addlegendentry{physical predictor}
\addplot table[x=t,y=corrected_entropy,col sep=comma]{data/entropy_correction_history.csv};
\addlegendentry{corrected}
\end{axis}
\end{tikzpicture}
\caption{Entropy correction diagnostic.  Positivity is preserved in all logarithmic updates, but the uncorrected reconstruction accumulates additional entropy from the high-frequency defect.}
\label{fig:entropy-correction}
\end{figure}

\begin{table}[htbp]
\centering
\scriptsize
\caption{Toy entropy-correction diagnostic at $T=1$.  The standard log reconstruction remains positive but accumulates entropy from a high-frequency reconstruction defect.  The corrected update damps only the defect relative to the physical logarithmic predictor.}
\label{tab:entropy-correction-summary}
\begin{tabular}{cccc}
\toprule
standard $\Phi$ & predictor $\Phi$ & corrected $\Phi$ & mean $\theta$ \\
\midrule
7.914403E-02 & 7.706680E-02 & 7.706680E-02 & 0.000999 \\
\bottomrule
\end{tabular}
\end{table}

The dynamic high-Weissenberg loop tests accumulation over many steps at $\Wi=80$ and background stretch $\Lambda=500$.  The standard logarithmic reconstruction remains positive throughout, but the entropy excess and stress-force defect accumulate.  The corrected reconstruction damps only the incompatible part and reduces the force defect by orders of magnitude.

\begin{figure}[htbp]
\centering
\begin{tikzpicture}
\begin{groupplot}[
  group style={group size=2 by 1, horizontal sep=1.2cm},
  width=0.47\textwidth,
  height=0.35\textwidth,
  grid=both,
  xlabel={$t$},
  legend style={at={(0.5,1.16)}, anchor=south, font=\scriptsize, draw=none, fill=none, /tikz/every even column/.append style={column sep=0.45em}},
  every axis plot/.append style={thick}
]
\nextgroupplot[ylabel={entropy excess}]
\addplot table[x=t,y expr=\thisrow{standard_entropy}-\thisrow{reference_entropy},col sep=comma]{data/dynamic_high_wi_history.csv};
\addlegendentry{standard log}
\addplot table[x=t,y expr=\thisrow{corrected_entropy}-\thisrow{reference_entropy},col sep=comma]{data/dynamic_high_wi_history.csv};
\addlegendentry{corrected}
\nextgroupplot[ylabel={stress-force defect}, ymode=log]
\addplot+[restrict x to domain=0.000001:1]
table[x=t,y=standard_force_defect,col sep=comma]{data/dynamic_high_wi_history.csv};
\addlegendentry{standard log}
\addplot+[restrict x to domain=0.000001:1]
table[x=t,y=corrected_force_defect,col sep=comma]{data/dynamic_high_wi_history.csv};
\addlegendentry{corrected}
\end{groupplot}
\end{tikzpicture}
\caption{Dynamic high-Weissenberg reconstruction loop.  Positive standard log reconstruction accumulates entropy and stress-force defects, while the corrected log reconstruction damps the incompatible defect at each step.}
\label{fig:dynamic-high-wi}
\end{figure}
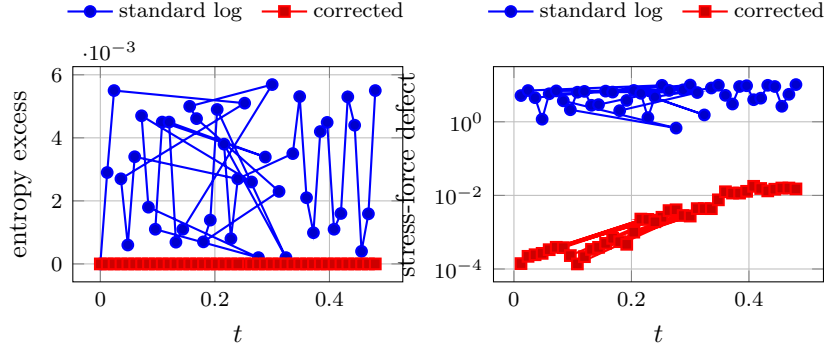

\begin{table}[htbp]
\centering
\caption{Dynamic high-Weissenberg reconstruction loop at $\Wi=80$ and background stretch $\Lambda=500$.  The standard logarithmic reconstruction remains positive but accumulates extra entropy and stress-force defects over many steps.  The corrected reconstruction damps only the entropy-incompatible defect.}
\label{tab:dynamic-high-wi-summary}
\begin{tabular}{ccccc}
\toprule
final $\Delta\Phi_{\rm std}$ & final $\Delta\Phi_{\rm corr}$ & max force std & max force corr & mean $\theta$ \\
\midrule
5.499871E-03 & -2.314380E-08 & 1.029331E+01 & 1.782488E-02 & 0.000198 \\
\bottomrule
\end{tabular}
\end{table}

\subsection{Resolution benchmark}

The final benchmark sweeps mesh size and stretch.  Direct physical-space reconstruction and raw log reconstruction both produce stress-force defects that grow with stretch unless the mesh is refined.  The corrected log reconstruction keeps the entropy-incompatible component close to the tolerance, while retaining the log-relative order on the resolved meshes.

\begin{table}[htbp]
\centering
\scriptsize
\caption{High-Weissenberg resolution benchmark at background stretch $\Lambda=500$.  The direct physical-space perturbation, raw log reconstruction, and corrected log reconstruction are compared over a mesh sweep.}
\label{tab:resolution-benchmark}
\begin{tabular}{cccccc}
\toprule
$N$ & direct min eig & force direct & force raw log & force corrected & $\theta$ \\
\midrule
16 & 9.147438E-01 & 2.053052E+00 & 2.080376E+00 & 3.174399E-05 & 0.000015 \\
24 & 9.138735E-01 & 1.187094E+00 & 1.202579E+00 & 7.339958E-05 & 0.000061 \\
32 & 9.135890E-01 & 7.274251E-01 & 7.369323E-01 & 8.995755E-05 & 0.000122 \\
48 & 9.133865E-01 & 3.430852E-01 & 3.475849E-01 & 8.485959E-05 & 0.000244 \\
64 & 9.133121E-01 & 1.969756E-01 & 1.995631E-01 & 9.744294E-05 & 0.000488 \\
\bottomrule
\end{tabular}
\end{table}

\begin{table}[htbp]
\centering
\scriptsize
\caption{Stretch sweep for the resolution benchmark on the $N=32$ grid.  Force defects grow with stretch for raw positive reconstruction, while the entropy-compatible correction keeps the incompatible component small.}
\label{tab:stretch-benchmark}
\begin{tabular}{ccccc}
\toprule
$\Lambda$ & force direct & force raw log & force corrected & entropy raw log \\
\midrule
1.000000E+02 & 1.455880E-01 & 1.476736E-01 & 3.605312E-05 & 9.100551E-06 \\
5.000000E+02 & 7.274251E-01 & 7.369323E-01 & 8.995755E-05 & 4.424172E-05 \\
1.000000E+03 & 1.454726E+00 & 1.473162E+00 & 1.798294E-04 & 8.780272E-05 \\
\bottomrule
\end{tabular}
\end{table}

The diagnostics support the main message of the analysis.  Positivity is necessary but not sufficient; the tensor entering the stress force must also be compatible with the tensor entering the entropy calculation.  The corrected log reconstruction supplies that compatibility while remaining inactive when the raw high-order reconstruction is already admissible.

\section{Implications for high-Weissenberg schemes}

The analysis yields several implementation rules.  First, eigenvalue positivity is monitored together with entropy and stress-work defects.  Second, the same accepted tensor is used in stress force, stretching, entropy variables, and quadrature.  Third, endpoint logarithmic spectral bounds give a practical way to control every state on the corrected path.  Fourth, a cellwise budget $\tau_K$ proportional to the local high-order truncation scale avoids global overdamping.  Finally, the largest admissible parameter is used so that the correction does not become unnecessarily dissipative.

For high-Weissenberg computations, the most useful practical diagnostic is not only $\min\lambda(\AA)>0$.  The reconstruction entropy excess and the stress-force defect must also be tracked.  A run can pass the positivity check and still contain nonphysical energy injection if the accepted tensor is not used consistently in the coupled work terms.

\section{Conclusion}

We have separated positivity preservation from entropy compatibility for high-Weissenberg conformation-tensor reconstruction.  Log and square-root variables keep $\AA$ in the positive definite cone, but positive tensors can still generate entropy bias, exponential amplification, and sign-indefinite work defects.  The corrected logarithmic reconstruction proposed here accepts the raw high-order log reconstruction only if it satisfies an elastic entropy budget; otherwise it selects the largest admissible point on the logarithmic path.  The method is positive, least damping, local, spectrally controlled, and compatible with a coupled velocity--pressure--conformation energy estimate.  The analysis proves asymptotic inactivity and conditional high-stretch resolution advantages, while the diagnostics verify that the correction removes nonphysical entropy and stress-force defects before they enter the momentum equation.

\clearpage
\appendix
\section{Scalar entropy bias}
The scalar elastic entropy is
\[
  \phi(a)=a-\log a-1,\qquad a>0.
\]
For a logarithmic perturbation $a=a_0\exp(\delta)$ with zero mean $\langle\delta\rangle=0$, Taylor expansion gives
\[
  \langle a\rangle=a_0\left(1+\frac12\langle\delta^2\rangle+\mathcal O(\|\delta\|_{L^\infty}^3)\right).
\]
The entropy mean satisfies
\[
  \langle\phi(a_0e^\delta)\rangle-\phi(a_0)
  =
  a_0\left(\frac12\langle\delta^2\rangle+\mathcal O(\|\delta\|_{L^\infty}^3)\right),
\]
because the logarithmic contribution is linear in $\delta$ and has zero mean.  Thus a strictly positive logarithmic reconstruction can increase elastic entropy even when the log perturbation has zero average.  This is a Jensen-type effect, not loss of positive definiteness.

For a square-root perturbation $a=a_0(1+\eta)^2$ with zero mean $\langle\eta\rangle=0$, one obtains
\[
  \langle a\rangle-a_0=a_0\langle\eta^2\rangle,
\]
and
\[
  \langle\phi(a_0(1+\eta)^2)\rangle-\phi(a_0)
  =
  a_0\langle\eta^2\rangle+\langle\eta^2\rangle+\mathcal O(\|\eta\|_{L^\infty}^3).
\]
Square-root reconstruction therefore has the same qualitative issue: positivity is preserved, but entropy can be added by the nonlinear map.

\section{Matrix exponential amplification}
For a symmetric logarithm $\Psi$ and perturbation $E$, the Frechet derivative of the exponential is
\[
  D\Exp_\Psi[E]=\int_0^1 e^{(1-s)\Psi}E e^{s\Psi}\,\dd s .
\]
Therefore
\[
  \|D\Exp_\Psi[E]\|_F
  \le e^{\lambda_{\max}(\Psi)}\|E\|_F .
\]
If $E$ commutes with $\Psi$ and is aligned with the largest eigendirection, this bound is sharp in order.  In high-stretch regions, where $\lambda_{\max}(\AA)=e^{\lambda_{\max}(\Psi)}$ is large, small log-space oscillations can be amplified into large physical stress-force defects.  The corrected log reconstruction does not remove this physical amplification entirely; rather, it prevents the amplified incompatible part from entering the entropy and work balance uncontrolled.

\section{Convexity along the logarithmic path}
Let
\[
  \Psi(\theta)=\Psi_0+\theta E,\qquad
  \AA(\theta)=\Exp(\Psi(\theta)),\qquad
  g(\theta)=\tr(\AA(\theta))-\tr(\Psi(\theta))-d .
\]
The map $\Psi\mapsto\tr(e^\Psi)$ is convex on symmetric matrices, and $-\tr\Psi-d$ is affine.  Hence $g$ is convex on $[0,1]$.  One direct verification uses the second variation
\[
\begin{aligned}
  \frac{\dd^2}{\dd\theta^2}\tr(e^{\Psi+\theta E})
  &=
  2\int_{0\le s\le r\le 1}
  \tr\!\left(
  e^{(1-r)(\Psi+\theta E)}
  E e^{(r-s)(\Psi+\theta E)}
  E e^{s(\Psi+\theta E)}
  \right)\dd s\,\dd r ,
\end{aligned}
\]
which is nonnegative by cyclicity and symmetry after diagonalization of $\Psi+\theta E$.  The entropy admissible set
\[
  \{\theta\in[0,1]: g(\theta)\le g(0)+\tau_h\}
\]
is therefore an interval containing zero.  The corrected reconstruction chooses its largest element, so it is the least-damping admissible point on the prescribed log path.

\section{Bisection gap and entropy budget}
Let $\theta_\star$ be the maximal admissible parameter.  Bisection maintains an interval
\[
  \theta_L^m\le \theta_\star < \theta_R^m,\qquad
  \theta_R^m-\theta_L^m=2^{-m}.
\]
The returned value $\theta_L^m$ is positive and entropy-compatible.  If $g$ is Lipschitz on $[0,1]$ with constant
\[
  L_g\le
  \sup_{\theta\in[0,1]}
  \left|\tr(D\Exp_{\Psi(\theta)}[E])-\tr E\right|,
\]
then
\[
  0\le g(\theta_\star)-g(\theta_L^m)\le L_g2^{-m}.
\]
Thus the numerical gap between the ideal least-damping parameter and the implemented parameter is controlled by the bisection depth.  In practice the correction should use an absolute entropy tolerance and a roundoff guard:
\[
  g(\theta)\le g(0)+\tau_h+c_{\rm mach}\epsilon_{\rm mach}(1+|g(0)|).
\]

\section{Work-defect decomposition}
Let $\AA_e$ be the tensor entering the entropy variable and let $\AA_m$ be the tensor entering the momentum stress work.  In a coupled step the polymeric work mismatch contains
\[
  \mathcal W_h
  =
  \frac{1-\beta}{\Wi}
  \sum_q w_q(\AA_m-\AA_e):\nabla\uu_h .
\]
The Cauchy--Schwarz and Young inequalities give
\[
  |\mathcal W_h|
  \le
  \frac{\beta}{4}\|\nabla\uu_h\|^2
  +
  \frac{(1-\beta)^2 C_Q^2}{\beta\Wi^2}
  \|\AA_m-\AA_e\|_Q^2 ,
\]
where $C_Q$ is the quadrature-to-volume stability constant.  Exact entropy compatibility corresponds to $\AA_m=\AA_e$, in which case this defect vanishes.  Positive but mismatched reconstructions are stable only up to this computable defect budget, and the budget is amplified by $\Wi^{-2}$ in the energy scaling.

\section{Cellwise correction}
The global correction selects one parameter for all cells, while the cellwise variant selects $\theta_K$ by
\[
  \sum_{q\in K}w_q\Phi(\Exp(\Psi_{0,q}+\theta_K E_q))
  \le
  \sum_{q\in K}w_q\Phi(\Exp\Psi_{0,q})+\tau_K .
\]
If $\sum_K\tau_K\le\tau_h$, summing over cells gives the same global entropy budget.  Cellwise correction is less dissipative because smooth cells can keep $\theta_K=1$ even when a small subset of cells requires damping.  The proof is exactly the interval argument above applied on each cell, followed by summation of the local inequalities.

\section{Asymptotic inactivity}
Assume the raw high-order log reconstruction satisfies
\[
  E_q=\mathcal O(h^{k+1})
\]
and the first variation is entropy-consistent,
\[
  \sum_q w_q\,D\Phi(\Exp\Psi_{0,q})[D\Exp_{\Psi_{0,q}}E_q]
  =\mathcal O(h^{2k+2}).
\]
The second-order expansion of the entropy profile gives
\[
  g(1)-g(0)=\mathcal O(h^{2k+2}).
\]
With a mesh-scaled budget $\tau_h=c_\tau h^{2k+2}$ and $c_\tau$ chosen above the leading consistency constant, the raw reconstruction is accepted for sufficiently small $h$.  In the marginal case, convexity implies
\[
  1-\theta_\star=\mathcal O(h^{k+1})
\]
provided $g'(1)$ is bounded below away from zero on the active branch.  Hence the correction is a high-order admissibility filter, not a first-order limiter.

\section{Conditional high-stretch resolution advantage}
Suppose the logarithmic field $\Psi=\Log\AA$ is smooth while the physical tensor has stretch $\Lambda=\lambda_{\max}(\AA)\gg1$.  A direct physical-space reconstruction of $\AA$ differentiates the exponential map and therefore carries factors of size $\Lambda$ in absolute tensor norms.  A log reconstruction approximates $\Psi$ directly, so its natural error is measured in log-relative or entropy metrics:
\[
  \|\Log\AA_h-\Log\AA\|,\qquad
  \int \Phi(\AA_h)-\Phi(\AA)-D\Phi(\AA):(\AA_h-\AA).
\]
The corrected log method has two advantages under the stated assumptions.  First, it keeps the approximation in the positive cone without clipping eigenvalues.  Second, it prevents the amplified component of a high-order log defect from appearing as a positive entropy increment or uncancelled work defect.  The statement is conditional because it compares regimes in which $\Psi$ is the resolved smooth variable; it is not a claim that log variables remove all high-Weissenberg stress-layer resolution requirements.

\section{Extended scalar and matrix diagnostics}
The figure below records one of the pointwise mechanisms used to motivate the corrected reconstruction.

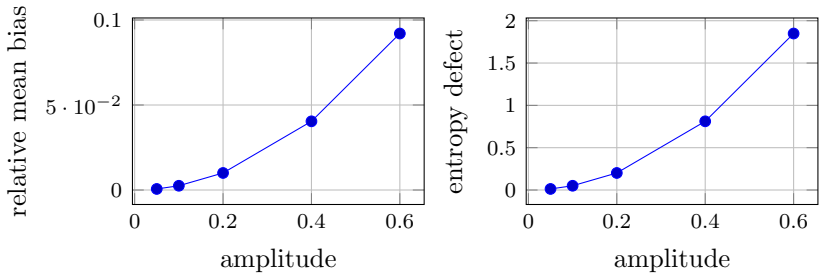
\begin{figure}[htbp]
\centering
\begin{tikzpicture}
\begin{groupplot}[
  group style={group size=2 by 1,horizontal sep=1.35cm},
  width=.43\textwidth,
  height=.32\textwidth,
  grid=major]
\nextgroupplot[xlabel={amplitude},ylabel={relative mean bias}]
\addplot+[mark=*] table[x=amplitude,y=relative_mean_bias,col sep=comma]{data/scalar_log_bias.csv};
\nextgroupplot[xlabel={amplitude},ylabel={entropy defect}]
\addplot+[mark=*] table[x=amplitude,y=entropy_defect,col sep=comma]{data/scalar_log_bias.csv};
\end{groupplot}
\end{tikzpicture}
\caption{Scalar logarithmic Jensen bias.  Both the mean physical stretch and the elastic entropy increase quadratically with the log perturbation amplitude.}
\label{fig:supp-scalar-bias}
\end{figure}

\section{Extended coupled diagnostics}
The coupled manufactured diagnostic measures the stress-force defect that would enter the momentum equation.  The dynamic high-Weissenberg loop records how a small positive reconstruction defect can accumulate over many steps.

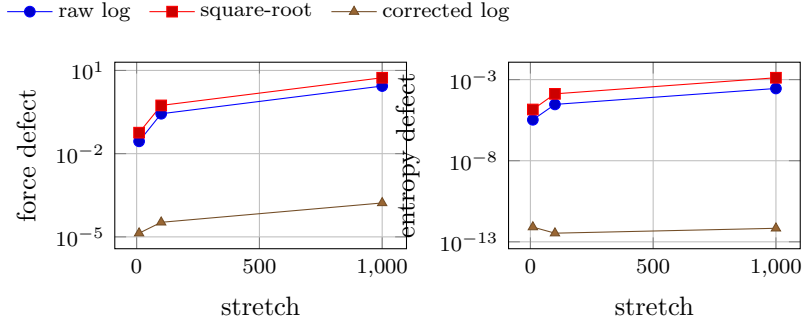
\begin{figure}[htbp]
\centering
\begin{tikzpicture}
\begin{groupplot}[
  group style={group size=2 by 1,horizontal sep=1.35cm},
  width=.43\textwidth,
  height=.32\textwidth,
  grid=major,
  legend style={at={(0.5,1.16)}, anchor=south, font=\scriptsize, draw=none, fill=none, /tikz/every even column/.append style={column sep=0.45em}}]
\nextgroupplot[xlabel={stretch},ylabel={force defect},ymode=log]
\addplot+[mark=*] table[x=stretch,y=log_force_defect,col sep=comma]{data/coupled_flow_diagnostic.csv};
\addlegendentry{raw log}
\addplot+[mark=square*] table[x=stretch,y=sqrt_force_defect,col sep=comma]{data/coupled_flow_diagnostic.csv};
\addlegendentry{square-root}
\addplot+[mark=triangle*] table[x=stretch,y=corrected_force_defect,col sep=comma]{data/coupled_flow_diagnostic.csv};
\addlegendentry{corrected log}
\nextgroupplot[xlabel={stretch},ylabel={entropy defect},ymode=log]
\addplot+[mark=*] table[x=stretch,y=log_entropy_defect,col sep=comma]{data/coupled_flow_diagnostic.csv};
\addplot+[mark=square*] table[x=stretch,y=sqrt_entropy_defect,col sep=comma]{data/coupled_flow_diagnostic.csv};
\addplot+[mark=triangle*] table[x=stretch,y=corrected_entropy_defect,col sep=comma]{data/coupled_flow_diagnostic.csv};
\end{groupplot}
\end{tikzpicture}
\caption{Coupled manufactured reconstruction diagnostic.  Positive raw log and square-root reconstructions both produce stretch-amplified force defects; the corrected log reconstruction removes the entropy-incompatible component.}
\label{fig:supp-coupled-defect}
\end{figure}

\begin{figure}[htbp]
\centering
\begin{tikzpicture}
\begin{groupplot}[
  group style={group size=2 by 1,horizontal sep=1.35cm},
  width=.43\textwidth,
  height=.32\textwidth,
  grid=major,
  legend style={at={(0.5,1.16)}, anchor=south, font=\scriptsize, draw=none, fill=none, /tikz/every even column/.append style={column sep=0.45em}}]
\nextgroupplot[xlabel={$t$},ylabel={entropy excess}]
\addplot+[mark=none] table[x=t,y expr=\thisrow{standard_entropy}-\thisrow{reference_entropy},col sep=comma]{data/dynamic_high_wi_history.csv};
\addlegendentry{raw log}
\addplot+[mark=none,dashed] table[x=t,y expr=\thisrow{corrected_entropy}-\thisrow{reference_entropy},col sep=comma]{data/dynamic_high_wi_history.csv};
\addlegendentry{corrected}
\nextgroupplot[xlabel={$t$},ylabel={force defect},ymode=log]
\addplot+[mark=none] table[x=t,y expr={max(\thisrow{standard_force_defect},0.000000000001)},col sep=comma]{data/dynamic_high_wi_history.csv};
\addplot+[mark=none,dashed] table[x=t,y expr={max(\thisrow{corrected_force_defect},0.000000000001)},col sep=comma]{data/dynamic_high_wi_history.csv};
\end{groupplot}
\end{tikzpicture}
\caption{Dynamic high-Weissenberg reconstruction loop at $\Wi=80$ and background stretch $\Lambda=500$.  The corrected reconstruction suppresses accumulation of the incompatible entropy and force defects.}
\label{fig:supp-dynamic-high-wi}
\end{figure}
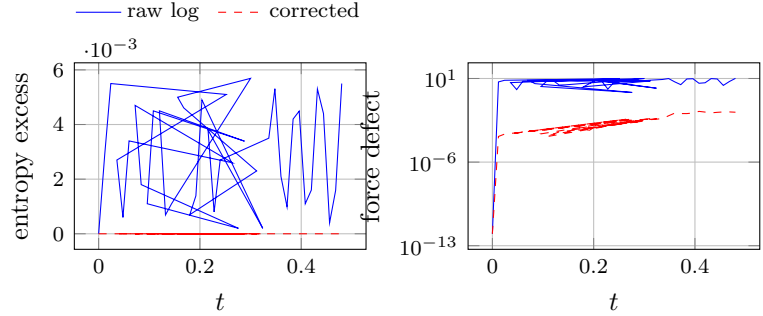

\section*{Acknowledgments}

The author acknowledges financial support from the National Natural Science Foundation of China (NSFC, Grant No. 12501602), the Education Department of Hunan Province (Grant No. 24C0055), the Science and Technology Department of Hunan Province (Grant No. 2025JJ60052), and the Scientific Research Start-up Fund of Xiangtan University (Grant No. KZ0810769).

\end{document}